\documentclass[american,aps,pra,reprint,superscriptaddress, twocolumn,longbibliography,floatfix]{revtex4-2}
\usepackage[dvips]{graphicx}
\usepackage{amsmath,amssymb,physics,amsthm,mathrsfs,amsfonts,dsfont,booktabs}
\usepackage{subfigure, epsfig}
\usepackage{braket}
\usepackage{bbm}
\usepackage{enumerate}
\usepackage{qcircuit}
\usepackage{algorithm}
\usepackage{algpseudocode}
\usepackage{hyperref}
\usepackage[svgnames]{xcolor}
\hypersetup{
     colorlinks = true,
     linkcolor = red,
     anchorcolor = maroon,
     citecolor = blue,
     filecolor = red,
     urlcolor = magenta
     }
\usepackage{comment}
\usepackage[normalem]{ulem}

\usepackage{bm}
\usepackage{amssymb,amsbsy,amsmath}
\usepackage{multirow,amssymb,amsbsy,amsmath}
\usepackage{graphicx}
\usepackage{verbatim}
\usepackage{float}

\usepackage{amsmath}

\newcommand{\mc}[1]{\mathcal{#1}}
\DeclareMathOperator{\sq}{\mathrm{sq}}

\usepackage{varioref}
\labelformat{equation}{#1}

\labelformat{section}{#1}

\labelformat{figure}{#1}

\labelformat{proposition}{#1}

\labelformat{lemma}{#1}

\labelformat{theorem}{#1}

\labelformat{observation}{#1}

\labelformat{definition}{#1}

\labelformat{problem}{#1}

\labelformat{algorithm}{#1}

\labelformat{corollary}{#1}

\labelformat{statement}{#1}

\newtheorem{theorem}{Theorem}

\newtheorem{proposition}{Proposition}

\newtheorem{lemma}{Lemma}

\newtheorem{corollary}{Corollary}
\newtheorem{definition}{Definition}

\begin{document} 

\title{Hysteretic squashed entanglement in many-body quantum systems}

\author{Siddhartha Das}\email{das.seed@iiit.ac.in}
\affiliation{q4i, Centre for Quantum Science and Technology (CQST), Center for Security, Theory and Algorithmic Research (CSTAR), International Institute of Information Technology Hyderabad, Gachibowli 500032, Telangana, India}

\author{Alexander Yosifov}\email{alexanderyyosifov@gmail.com}
\affiliation{Clarendon Laboratory, University of Oxford, Parks Road, Oxford, OX1 3PU, United Kingdom}
\affiliation{School of Physical and Chemical Sciences, Queen Mary University of London, London E1 4NS, United Kingdom}

\author{Jinzhao Sun}\email{jinzhao.sun.phys@gmail.com}
\affiliation{School of Physical and Chemical Sciences, Queen Mary University of London, London E1 4NS, United Kingdom}

\date{\today}

\begin{abstract}

Entanglement in many-body quantum systems is distributed across spatial regions, where its structure often dictates the information-processing capabilities of the state. Yet, characterizing the entanglement structure, especially for mixed states, remains a challenge. In this work, we propose hysteretic squashed entanglement $T_{\sq}$, a conditional entanglement monotone that measures the genuine quantum correlations between two subregions, conditioned on a third region, in a many-body quantum state. $T_{\sq}$ is upper bounded by the convex-roof extension of quantum conditional mutual information and exhibits several desirable properties like monogamy, convexity, asymptotic continuity, faithfulness, and additivity for tensor-product states. We study the conditional entanglement generation in a one-dimensional transverse-field Ising model under quench, where we show that $T_{\sq}$ effectively squashes classical contributions and can detect genuine quantum correlations across both adjacent and long-range subsystems. We elucidate the utility of this measure as a robust quantifier of topological entanglement entropy for mixed states. This opens new operational resource-theoretic avenues for probing topological order and criticality.

\end{abstract}

\maketitle

Characterizing how quantum correlations are distributed across spatial regions is a major problem in quantum many-body physics \cite{horodecki2009quantum,amico2008entanglement,zeng2015quantum}, with applications to computation, communication, error correction, and metrology \cite{d2001using,DBWH21,PhysRevLett.70.1895,DW19,brun2006correcting,PhysRevA.52.R2493,pezze2021entanglement}. However, in extended open systems correlations between two regions are typically mediated by surrounding parties and mixed with classical contributions. Scale also plays a role here as short-range quantum correlations among adjacent regions coexist with long-range ones between distant regions that are relayed through intermediate subsystems. This raises a fundamental question: how to quantify the genuine quantum correlations between two spatial regions, conditioned on an intermediate region that can mediate information? Addressing this is essential for understanding non-equilibrium quantum dynamics and emergent order in many-body systems~\cite{zeng2015quantum}.

A natural starting point is the quantum conditional mutual information (QCMI)~\cite{LR73} which has found many applications in informational aspects of many-body quantum systems~\cite{zeng2015quantum,svetlichnyy2022decay,Lev24,DGP24}. The QCMI $I(A;C|B)_{\rho}$ of a quantum state $\rho_{ABC}$ quantifies the total correlations between regions $A$ and $C$, conditioned on $B$. Later, that became the basis for squashing out correlations in a state of a composite system that can be relayed by some compatible state extension to determine intrinsic or genuine quantum correlations present in the system. %Later, that became the basis for squashing, where correlations that can be reproduced by conditioning on a compatible extension system E are regarded as non-intrinsic and are hence “squashed”.
This idea underlies squashed entanglement $E_{\sq}$ \cite{CW04} that measures bipartite entanglement and its generalization squashed quantum non-Markovianity $N_{\sq}$~\cite{GPG+25}, used for studying genuine non-Markovianity in tripartite quantum states and information backflow in open quantum systems~\cite{BGG+25}. Despite recent success in defining faithful and operationally meaningful measures of quantum correlations, current measures cannot distinguish quantum correlations that are genuinely shared between distant regions from those that are conditionally relayed through intermediate region in a many-body quantum state. In other words, there is no information-theoretic measure that quantifies the irreducible conditional quantum correlations between spatial subregions of a many-body state~\cite{MRCI23}. This gap is especially pronounced in mixed state regimes, where mediation and mixing are basic features of the correlation structure.

Here, we introduce \textit{hysteretic squashed entanglement} $T_{\sq}$, which for a (quadripartite) many-body state quantifies the entanglement between two subregions, conditioned on the third and any compatible extension system, while the fourth subregion remains a silent spectator. More formally, $T_{\sq}(A;C|B)_{\rho}$ of a state $\rho_{ABCD}$ quantifies the conditional entanglement between $A$-$C$, when $D$ is silent (excluded), that cannot be relayed through an intermediate conditioning region $B$ which may include any compatible extension system [see Eq.~\eqref{def:df1}]. And importantly, it is distinct from both bipartite and genuine multipartite entanglement (GME). Even though $T_{\sq}$ (similar to $E_{\sq}$ and $N_{\sq}$) is numerically difficult to compute in general, we show that its utility lies in its properties that can shed light on the structure of the states and possible physical ramifications for many-body systems. 

We show that $T_{\sq}$ satisfies several desirable properties of a conditional entanglement monotone~\cite{YHW08,BGG+25}: convexity, faithfulness, asymptotic continuity, additivity, and monogamy. These properties clarify the intrinsic structure of a many-body quantum state when its $T_{\sq}$ is approximately vanishing. Furthermore, $T_{\sq}$ can provide a robust diagnostic of topological entanglement entropy (TEE) in mixed states~\cite{KP06,LW05,LW06}, and lays the groundwork for a resource-theoretic treatment of topological order beyond pure states~\cite{WSMG25,EC25}. As another application, we perform numerical study of the conditional entanglement generation in a quenched 1D Ising model, where $T_{\sq}$ detects quantum correlations in regimes where GME-specific diagnostics vanish.

Before we present our main results, we remark that $T_{\sq}$ admits an operational interpretation within a secure communication protocol adapted from the one-time quantum conditional pad~\cite{SWW20,GPG+25}. Consider a scenario where multiple copies of a state $\rho_{ABCD}$ are distributed such that Alice holds $A$, Charlie holds $C$, and an eavesdropper Eve possesses $B$ and all compatible extensions $E$, but lacks access to $D$. Alice and Charlie are connected by an ideal quantum channel whose output is also accessible to Eve. In this setting, $T_{\sq}(A;C|B)_{\rho}$ is exactly twice the optimal rate at which Alice can send private messages to Charlie. The proof follows the same construction as \cite{SWW20} (see also Theorem 4 in \cite{GPG+25}).

\textit{Preliminaries.---} We briefly introduce some measures of interest. For a tripartite state $\rho_{ABC}$, QCMI is defined as $I(A;C|B)_\rho := I(A;BC)_\rho - I(A;B)_{\rho}$, where $I(A;B)_\rho := S(A)_\rho + S(B)_\rho - S(AB)_\rho$ is the quantum mutual information (QMI) and $S(\cdot) := -\mathrm{tr}[\cdot \log \cdot]$ is the von Neumann entropy.
% , and $\log$ is taken w.r.t.~base $2$. 
The nonnegativity of QCMI is equivalent to the strong subaddivity of the entropy~\cite{LR73}, where its strengthening shows that for an arbitrary state $\rho_{ABC}$, there exists a universal recovery map $\mathcal{R}_{B\to BC}$~\cite{JRS+18}, such that $I(A;C|B)_{\rho}\geq -\log F(\rho_{ABC},\mathcal{R}_{B\to BC}(\rho_{AB}))\geq 0$, where $F(\rho,\sigma):=\norm{\sqrt{\rho}\sqrt{\sigma}}_{1}^2$ is the fidelity between states $\rho,\sigma$. For a bipartite state $\rho_{AB}$, squashed entanglement is defined as $E_{\sq}(A;C):=\frac{1}{2}\inf\{I(A;C|E)_{\rho}:\tr_E[\rho_{ACE}]=\rho_{AC}\}$. For a tripartite state $\rho_{ABC}$, squashed quantum non-Markovianity is defined as $N_{\sq}(A;C|B)_\rho := \tfrac{1}{2} \inf \{ I(A;C|BE)_\rho : \mathrm{tr}_E[\rho_{ABCE}] = \rho_{ABC}\}$.

\textit{Properties.---} We now formally define our measure.

\begin{definition}
\label{def:df1}
    For a relevant quadripartite density operator $\rho_{ABCD}$, its (squashed) hysteretic squashed entanglement $T_{\sq}(A;C|B)_{\rho}$ reads
\begin{equation}\label{eq:sqtee}
    T_{\sq}(A;C|B)_{\rho}=\frac{1}{2}\inf_{\substack{\sigma_{ABCDE}:\\ \tr_E(\sigma_{ABCDE})=\rho_{ABCD}}}I(A;C|BE)_{\sigma},
\end{equation}
where the optimization is with respect to all possible state extensions $\sigma_{ABCDE}$ of $\rho_{ABCD}$. 
\end{definition}

It follows from Definition~\ref{def:df1} that $T_{\sq}(A;C|B)_{\rho}\leq \log\min\{|A|,|C|\}$ for a state $\rho_{ABCD}$ and the inequality is saturated iff its marginal state $\rho_{AC}$ is a maximally entangled pure state of Schmidt rank $d=\min\{|A|,|C|\}$. We note that to evaluate $T_{\sq}$ with finite-dimensional $ABCD$, the infimum can be restricted (without loss of generality) to finite-dimensional $E$; whereas for infinite-dimensional $ABCD$, $E$ is also infinite-dimensional~\cite{CW04,BCY11}. 
In this work, we focus on many-body systems composed of finite-dimensional regions. Physically, $T_{\sq}$ quantifies the irreducible $A$-$C$ correlations that persist under conditioning on $B$ and any compatible $E$, and disregarding $D$; meaning that $T_{\sq}=0$ iff the $A$-$C$ correlations can be mediated through $B$ and some compatible $E$ that is hidden. 

$T_{\sq}$ exhibits the following desirable properties.

\begin{proposition}
\label{thm:prop}
    For a many-body quantum system $ABCD$ composite of four regions, $T_{\sq}$ satisfies the following properties:\\~\\
    1) Convexity: For a state $\rho_{ABCD}$ expressed as a convex mixture $\sum_{x}p_X(x)\rho^x_{ABCD}$ of quantum states
        \begin{equation}
            T_{\sq}(A;C|B)_{\rho}\leq \sum_xp_X(x)T_{\sq}(A;C|B)_{\rho^x}.\\
        \end{equation}
    2) Faithfulness: For a state $\rho_{ABCD}$, $T_{\sq}(A;C|B)_{\rho}=0$ iff there exists a state extension $\rho_{ABCDE}$ and a universal recovery map $\mathcal{R}_{BE\to BCE}$, such that 
    \begin{equation}
    \mathcal{R}_{BE\to BCE}(\rho_{ABE})=\tr_D[\rho_{ABCDE}].\\
    \end{equation}
    3) Asymptotic continuity: For two finite-dimensional states $\rho_{ABCD}$ and $\sigma_{ABCD}$ close in trace distance, $\frac{1}{2}\norm{\rho-\sigma}\leq \varepsilon\in[0,1]$, we have
        \begin{equation}
            \abs{T_{\sq}(A;C|B)_{\rho}-T_{\sq}(A;C|B)_{\sigma}}\leq f(d,\varepsilon),
        \end{equation}
      where $f(d,\varepsilon)=2\sqrt{\varepsilon}\log d+(1+2\sqrt{\varepsilon})h_2\left(\frac{2\sqrt{\varepsilon}}{1+2\sqrt{\varepsilon}}\right)$ for $d=\min\{|A|,|C|\}$ and Shannon's binary entropy $h_2(p):=-p\log p-(1-p)\log (1-p)$, and $\lim_{\varepsilon\to 0^+}f(d,\varepsilon)=0$. \\
    4) Monogamy: For a quantum state $\rho_{A_1A_2BCD}$ 
        \begin{equation}
        T_{\sq}(A_1A_2;C|B)_{\rho}\geq T_{\sq}(A_1;C|B)_{\rho}+T_{\sq}(A_2;C|B)_{\rho}.\\
        \end{equation}
    5) Additive under tensor-product states: For a quantum state $\rho_{A_1A_2B_1B_2C_1C_2D_1D_2}$
        \begin{equation}
         T_{\sq}(A_1;C_1|B_1)_{\rho}+  T_{\sq}(A_2;C_2|B_2)_{\rho}\geq T_{\sq}(A;C|B)_{\rho},
        \end{equation}
        and the inequality saturates if $\rho_{A_1A_2B_1B_2C_1C_2D_1D_2}=\sigma_{A_1B_1C_1D_1}\otimes\tau_{A_2B_2C_2D_2}$, i.e.,
        \begin{equation}
        T_{\sq}(A;C|B)_{\rho}= T_{\sq}(A_1;C_1|B_1)_{\sigma}+ T_{\sq}(A_2;C_2|B_2)_{\tau},
        \end{equation}
        assuming $A=A_1A_2$, $B=B_1B_2$, $C=C_1C_2$.
\end{proposition}

Proposition~\ref{thm:prop} comes from Lemma 2 of~\cite{GPG+25}. Here: 1) ensures probabilistic mixing cannot create hysteretic squashed entanglement from states with $T_{\sq}=0$; 2) means $T_{\sq}=0$ exactly for states without irreducible conditional entanglement, and $T_{\sq}>0$ otherwise; 3) guarantees robustness because the less distinguishable two states are, the closer their $T_{\sq}$ values must be; 4) shows that conditional entanglement cannot be freely shared across multiple parties; while 5) implies $T_{\sq}$ does not increase by simply adjoining two states but equals the sum of their individual values. The consequence of 4) is also observed numerically in Fig.~\ref{fig:fig1}. 

Overall, Proposition~\ref{thm:prop} shows that $T_{\sq}$ captures a physically consistent notion of conditional entanglement in extended systems: it is stable under noise, cannot be freely shared among subsystems, and characterizes when global correlations can be reconstructed from local data. For instance, if $\rho_{ABCD}$ is a genuinely entangled pure state, then there exist $A$-$C$ correlations that are mediated, as well as ``shielded'' by $B$ and $D$; entanglement is hence shared among all the subsystems $ABCD$ without any leakage to the environment. 

\begin{table}[htbp]
\centering
\begin{tabular}{lcc}
\toprule
 & GHZ state ($\Phi_n$) & W state ($\Psi_n$) \\
\midrule
$n=3$           & 1 & 0.918 \\
$n=4$           & 0 & 0.377 \\
$n=5$           & 0 & 0.249 \\
$n \ge 4$       & 0 & $2h_2(\frac{2}{n}) - h_2(\frac{1}{n}) - h_2(\frac{3}{n})$ \\
\bottomrule
\end{tabular}
\caption{Comparison of $I(A;C|B)$ for $n$-party GHZ and W states. $h_2(p)$ denotes the binary entropy.}
\label{tab:qcmi_ghz}
\end{table}
To further elaborate, let $\Phi_n$ and $\Psi_n$ denote, respectively, $n$-qubit GHZ and W states~\cite{DVC00,HSD15,DBWH21}, and let this system be partitioned as $A$-$B$-$C$-$D$ for $n\geq 4$, where $A,B,C$ are each single qubit and the remaining qubits are in $D$. When $n=3$ for both GHZ and W states, $D$ is in a product state with $ABC$. Since $\Phi_n$ and $\Psi_n$ are pure, we get $T_{\sq}(A;C|B)_{\Phi_n}=\frac{1}{2}I(A;C|B)_{\Phi_n}$ and $T_{\sq}(A;C|B)_{\Psi_n}=\frac{1}{2}I(A;C|B)_{\Psi_n}$. The values of $I(A;C|B)$, which are double that of the respective $T_{\sq}$ for $n$-partite GHZ and W states, are shown in Table~\ref{tab:qcmi_ghz}. This is consistent with the fact that each two-qubit marginal of the W state remains entangled but tracing out even a single qubit from the GHZ state makes it separable. In contrast to the GHZ state, the QCMI for the W state remains strictly positive for any finite $n$, irrespective of how $n$-qubit W and GHZ states are partitioned into $A$-$B$-$C$-$D$. This non-vanishing $I(A;C|B)_{\Psi_n}$ reflects the intrinsic nature of W-type entanglement, which is spread across all pairs and triplets in a way that precludes reconstruction from local marginals. Even upon tracing out the observer qubit $D$, the resulting mixed state $\rho_{ABC}$ retains bipartite entanglement across all remaining partitions. Consequently, the intermediary region $B$ fails to act as a perfect shield (or a quantum Markov buffer) between regions $A$ and $C$.

Finally, we prove $T_{\sq}$ generalizes to $E_{\sq}$ and $N_{\sq}$ for relaxed spatial constraints (see \hyperref[sec:SM]{Supplemental Materials}).

\begin{lemma}\label{lem:lower_bounds}
     Suppose $\rho_{AC}$ and $\rho_{ABC}$ are reduced states of $\rho_{ABCD}$. The squashed quantum correlations $E_{\sq}$ of $\rho_{AC}$, $N_{\sq}$ of $\rho_{ABC}$, and $T_{\sq}$ of $\rho_{ABCD}$ can then be ordered as
\begin{align}
E_{\sq}(A;C)_{\rho}\leq N_{\sq}(A;C|B)_{\rho} \nonumber
&\leq T_{\sq}(A;C|B)_{\rho}\\ 
&\leq \frac{1}{2}I(A;C|D)_{\rho}. 
\end{align}
\end{lemma}

This formalizes a hierarchy relation in the context of constrained optimization; as more spatial regions are included $AC\rightarrow ABC\rightarrow ABCD$, the set of valid extensions $\sigma_{ABCD}\subseteq \sigma_{ABC}\subseteq \sigma_{AC}$ shrinks, where $\sigma_{X}$ denotes the set of state extensions $\rho_{XE}$, such that $\mathrm{tr}_{E}[\rho_{XE}]=\rho_{X}$. From the definitions of $E_{\sq}$ and $N_{\sq}$~\cite{CW04,GPG+25,BGG+25}, it is easy to see that an extension compatible with the global state $\rho_{ABCD}$ is automatically compatible with the marginals $\rho_{ABC}$ and $\rho_{AC}$. 

Note that $T_{\sq}$ is related to $N_{\sq}$ but formally distinct as the latter is defined for tripartite state while the former for a quadripartite state. There is no silent spectator $D$ present in the context of $N_{\sq}$ which is crucial in quantifying hysteretic entanglement between two subregions $A$-$C$ relative to $B$ in a quantum state $\rho_{ABCD}$, see Table~\ref{tab:qcmi_ghz}. 

\textit{Topological order of mixed states.---} It is well known~\cite{zeng2015quantum,Lev24} that for a gapped quantum many-body system in topologically ordered pure ground state $\psi_{ABCD}$, where the configuration $A$-$B$-$C$ follows the Levin-Wen (LW) scheme, the TEE is $\gamma=\frac{1}{2}I(A;C|B)_{\psi}$~\cite{KP06,LW05,LW06}. Here, $\gamma=0$ corresponds to a quantum Markov chain. Equivalently, $\gamma=0$ for $\psi_{ABCD}$ if there exists a universal recovery map $\mc{R}_{B\to BC}$ that can perfectly recover the reduced state $\psi_{ABC}$ from $\psi_{AB}$. For a mixed state $\rho_{ABCD}$, however, correlations may be mediated through an intermediate party, and an extension system $E$ must be introduced. Here, $\gamma=0$ is related to the existence of a state extension $\rho_{ABCDE}$ and a universal recovery map $\mc{R}_{BE\to BCE}$, such that $\mc{R}(\rho_{ABE})=\rho_{ABCE}$; the regions here are configured as $A$-$BE$-$C$. The convex-roof extension of QCMI (dubbed co(QCMI)), see Eq.~\eqref{eq:coqcmi}, was recently used~\cite{WSMG25} to study topological order for mixed states. 

$T_{\sq}(A;C|B)_{\rho}$ is a well-founded diagnostic of topological order in terms of TEE of a mixed state due to desirable properties of the conditional entanglement it satisfies. For pure state $\psi_{ABCD}$, $T_{\sq}(A;C|B)_{\psi}={\rm co(QCMI)[\psi_{ABCD}]}=\frac{1}{2}I(A;C|B)_{\psi}=\gamma$. For relevant many-body system in a mixed state $\rho_{ABCD}$, where regions $A,B,C$ are as per LW scheme, we refer to $T_{\sq}(A;C|B)_{\rho}$ as the hysteretic TEE. The following theorem states that $T_{\sq}$ is upper bounded by co(QCMI). We provide detailed proof in Appendix \hyperref[app:C]{C}.

\begin{theorem}
\label{thm:thm1}
    For an arbitrary relevant state $\rho_{ABCD}$, the hysteretic TEE $T_{\sq}(A;C|B)_{\rho}$ is upper bounded by the co(QCMI)
\begin{equation}
   0\leq T_{\sq}(A;C|B)_{\rho}\leq {\rm co(QCMI)}[\rho_{ABCD}].
   \label{eq:coqcmi2}
\end{equation}
\end{theorem}

Importantly, from here we can form a resource-theoretic approach~\cite{CG19,GPG+25} to study topological order in many-body mixed states. $T_{\sq}=0$ for states with no topological order, and those can be deemed ``free'', whereas states with non-zero topological order have $T_{\sq}>0$ and can be deemed ``resourceful''. The set of free operations includes local operations and classical communication~\cite{CLM+14} between non-conditioning regions, local operations and one-way quantum communication from non-conditioning to conditioning regions, and local operations on conditioning regions (proof follows from Lemma 3 of \cite{GPG+25}). Since $T_{\sq}$ is non-increasing under these operations, it serves as a valid monotone for quantifying topological order in many-body mixed states. In contrast, the tripartite measure $N_{\sq}$ cannot directly characterize TEE, as it lacks the spatial conditioning structure required to isolate correlations across separated regions.

\begin{theorem}
\label{thm:thm2}
     If $T_{\sq}(A;B|C)_{\rho}$ of a quantum state $\rho_{ABCD}$ changes negligibly under the action of a quantum channel $\mc{E}_{A\to A'}$ on $A$ then there exists a state extension $\rho_{ABCDE}$ and a universal recovery map $\mc{R}_{\rho_{ABE}\otimes\rho_C,\mc{E}}$ that almost recovers the state $\rho_{ABCE}$ from $\mc{E}(\rho_{ABCE})$. In particular,
     \begin{equation}
         T_{\sq}(A;C|B)_{\rho}-T_{\sq}(A';C|B)_{\mc{E}(\rho)}\leq \varepsilon 
     \end{equation}
     implies the existence of a quantum channel $\mc{R}_{\rho_{ABE}\otimes\rho_C,\mc{E}}$ such that 
     \begin{equation}
        \sup_{\substack{\rho_{ABCE}:\\ \tr_E(\rho_{ABCDE})=\rho_{ABCD}}} F(\rho_{ABCE},\mc{R}\circ\mc{E}(\rho_{ABCE}))\geq 2^{-\varepsilon}, 
     \end{equation}
     and $\mc{R}_{\rho_{ABE}\otimes\rho_C,\mc{E}}\circ\mc{E}(\rho_{ABCE})\approx\rho_{ABCE}$ iff
     $\varepsilon\approx 0$.
\end{theorem}
Theorem~\ref{thm:thm2} (detailed proof in Appendix \hyperref[app:C]{C}) shows that the correlations quantified by $T_{\sq}$ are stable under local processing: if a local operation changes $T_{\sq}$ only slightly, then there exists a state extension and a recover map that allow to approximately undo the action of the local processing.

\textit{Model.---} To illustrate the operational utility of $T_{\sq}$, we study the entanglement generation in an $N=8$ qubit 1D transverse-field Ising model with periodic boundary conditions, described by the Hamiltonian
\begin{equation}
\label{eq:TFIM}
H=-J\sum_{\langle ij\rangle}Z_{i}Z_{j}+h(t)\sum_{i}X_{i},
\end{equation}
where $J=1$ is the nearest-neighbor (NN) interaction, $Z_{i}$ and $X_{i}$ are the Pauli operators on $i$, and $h(t)$ is the time-dependent transverse magnetic field on all $i$. The system is initialized in $\ket{\psi_{0}}=\bigotimes_{i}\ket{\downarrow}_{i}$, corresponding to one of the two degenerate ferromagnetic ground states at $h=0$. We partition the chain into three overlapping regions $A,B,C$ via the LW scheme, thereby forming a ring. The remaining degrees of freedom constitute a fourth region $D$ that remains uncorrelated with $ABC$ throughout the evolution; hence, $\rho_{ABCD}(t)=\rho_{ABC}(t)\otimes\rho_D$ $\forall t$, in which case $N_{\sq} = T_{\sq}$. Note that for $ABC$ in a pure state $\psi_{ABC}$, $T_{\sq}(A;C|B)_{\rho}=\frac{1}{2}I(A;C)_{\psi}$. We then consider a quench under Eq.~\eqref{eq:TFIM}, where the system, coupled to a dephasing environment, is driven across the phase transition. $h(t)$ follows a standard linear ramp protocol: it increases from $0$ to $h_{\text{max}}=2J$ over time $t_\text{up}=150$, then it is held constant for $t_\text{hold}=20$, and ramped back to zero. The evolution is governed by the Lindblad master equation for the density matrix $\rho(t)$: $\dot{\rho} = -i [H(t), \rho] + \gamma \sum_{i} \left( Z_i \rho Z_i - \frac{1}{2} \{Z_{i}^{2}, \rho\} \right)$, where $\gamma$ is the local dephasing rate. The equation is solved via the Monte Carlo wavefunction method.

\textit{Numerical results.---} In Fig.~\ref{fig:fig1} we compare $T_{\sq}$ against the TMI $I_3(A:B:C) = S(A) + S(B) + S(C) - S(AB) - S(BC) - S(AC) + S(ABC)$; $I_{3}=0$ when the three subsystems are separable across any bipartition and $I_{3}>0$ in the presence of correlations (classical and quantum) shared among all three parties. At the onset of the quench $t\ll t_\text{up}$, both $T_\text{sq}$ and $I_{3}$ are near-zero, suggesting correlations are strictly local and confined to the NN $ZZ$ couplings in the ground state $\ket{\psi_{0}}$, Fig.~\ref{fig:fig1}(a). As $h(t)$ ramps up, quantum fluctuations from the non-commuting $X_i$ terms induce a rapid buildup of multipartite correlations, converting local $ZZ$ order into delocalized quasiparticle pairs.
Overall, we observe $T_\text{sq}$ remains systematically lower than $I_{3}$ due to the squashing of classical redundancies from the environment. This is consistent with the visible gap in Fig.~\ref{fig:fig1}(c) that only widens as the state becomes less pure under $\gamma$. Meanwhile, the macroscopic response to the quench is captured by the order parameters in Fig.~\ref{fig:fig1}(b): we observe the longitudinal magnetization $\langle Z \rangle$ decays sharply during the ramp-up as quantum fluctuations suppress long-range $ZZ$ order, while the transverse magnetization $\langle X \rangle$ peaks near criticality, signaling the transition to the paramagnetic regime. This is corroborated by the average NN correlations in Fig.~\ref{fig:fig1}(d), whose rapid decrease coincides with the decay of $\langle Z \rangle$ and the growth of both $T_{\text{sq}}$ and $I_3$. Physically, as $h(t)$ disrupts the initial spin alignment, the suppression of local ferromagnetic order gives rise to delocalized quantum fluctuations that drive tripartite entanglement. The close agreement between $T_{\text{sq}}$ and $I_3$ early on indicates the presence of predominantly quantum correlations, with classical contributions becoming significant only after prolonged dephasing.

\begin{figure}
    \centering
    \includegraphics[width=1.0\linewidth]{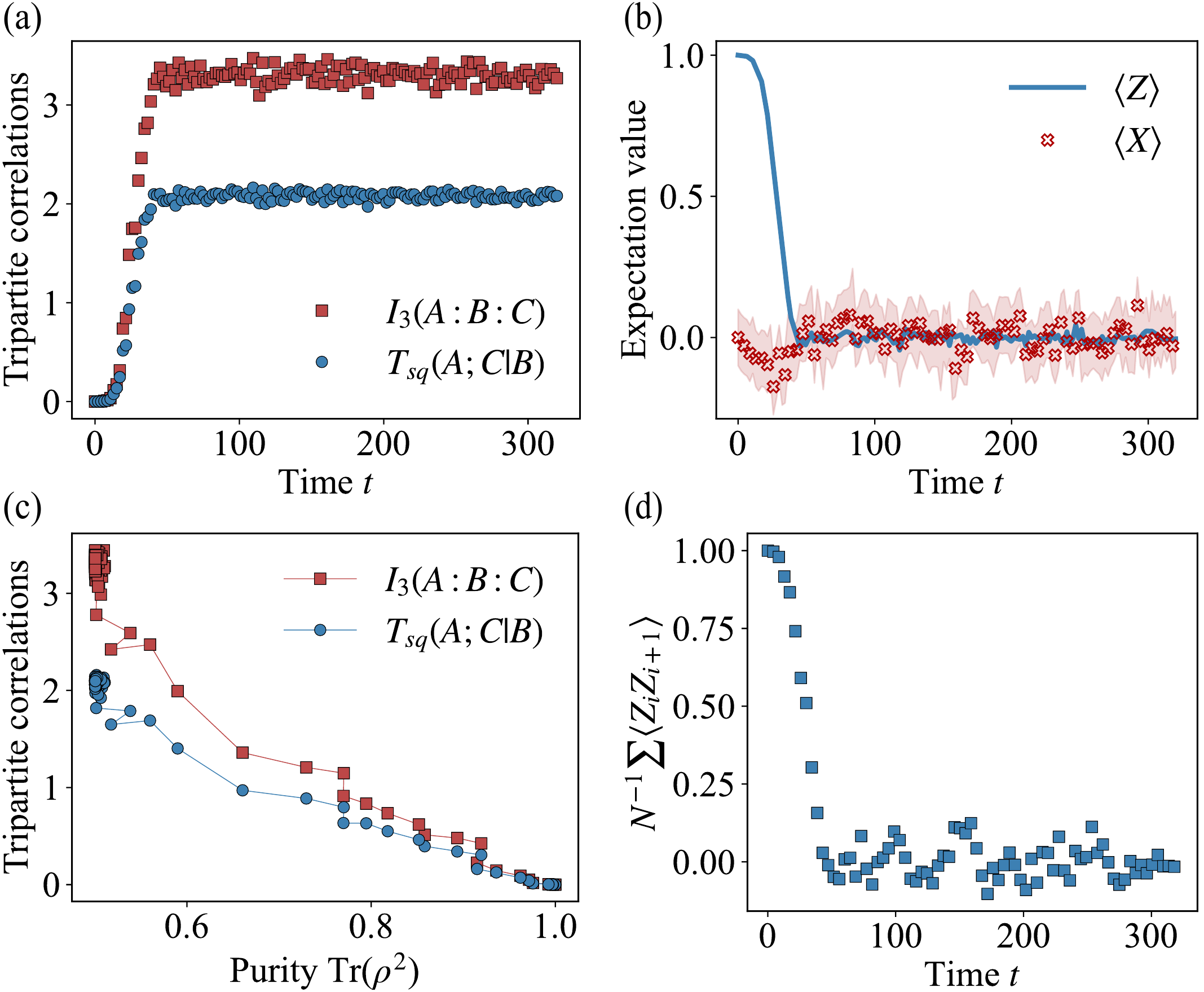}
    \caption{(a) Generation of tripartite correlations for the initial state $\ket{\psi_{0}}$ quenched under Eq.~\eqref{eq:TFIM} with $\gamma=0.5$. (b) As $h(t)$ ramps up, the system crosses the critical point, evident from the rapid suppression of the longitudinal magnetization order $\langle Z\rangle = N^{-1}\sum_i\langle Z\rangle$, hence marking the transition to a paramagnetic regime. (c) As purity decreases under $\gamma$, the gap widens, showing $T_{\text{sq}}$ effectively squashes classical correlations. (d) The average NN correlation function over time.}
    \label{fig:fig1}
\end{figure}

We now compare $T_{\text{sq}}$ with $\tau_{3}$, Fig.~\ref{fig:fig2}. Calculating its exact Coffman-Kundu-Wootters form \cite{PhysRevA.61.052306} for mixed states, however, requires convex roof optimization over pure-state decompositions which is intractable for general dissipative dynamics \cite{lohmayer2006entangled}. For this reason, we employ the negativity-based witness $\tau_{3}=\sqrt[3]{N_{A|BC}\cdot N_{B|AC}\cdot N_{C|AB}}$, where $N_{i|jk} = \frac{\|\rho_{ABC}^{T_i}\|_{1}-1}{2}$ is the negativity across the $i|jk$ bipartition and $\rho_{ABC}^{T_i}$ is the partial transpose with respect to subsystem $i$ \cite{jungnitsch2011taming, adesso2006multipartit}. $\tau_{3}$ provides a strict lower bound for GME, vanishes if the state is biseparable across any bipartition, and is efficiently computable for mixed states via direct diagonalization of the reduced density matrix. We see from Fig.~\ref{fig:fig2}(a) that adjacent qubits support both bipartite and GME. At early times $t \ll t_{\text{up}}$, the sharp increase in $\tau_{3}$ indicates the presence of irreducible GHZ-type correlations induced by NN $ZZ$ couplings. Its subsequent collapse to $0$ is expected and underscores the known fragility of GME under local $\gamma$ \cite{simon2002robustness}. In contrast, $T_{\sq}$ remains positive and oscillates, capturing persistent conditional entanglement. While for distant qubits, we see $\tau_{3}=0$, confirming that long-range correlations, generated by low-energy quasiparticle excitations during the quench, are bipartite, Fig.~\ref{fig:fig2}(b). $T_{\sq}$ is again positive here. Its oscillations indicate coherent quasiparticle recurrences within the lattice. Overall, while $\tau_{3}$ is only sensitive to residual GME, $T_{\sq}$ captures genuine conditional quantum correlations shared across a given bipartition, capturing all quantum resources on local and non-local scales.

\begin{figure}
    \centering
    \includegraphics[width=1.0\linewidth]{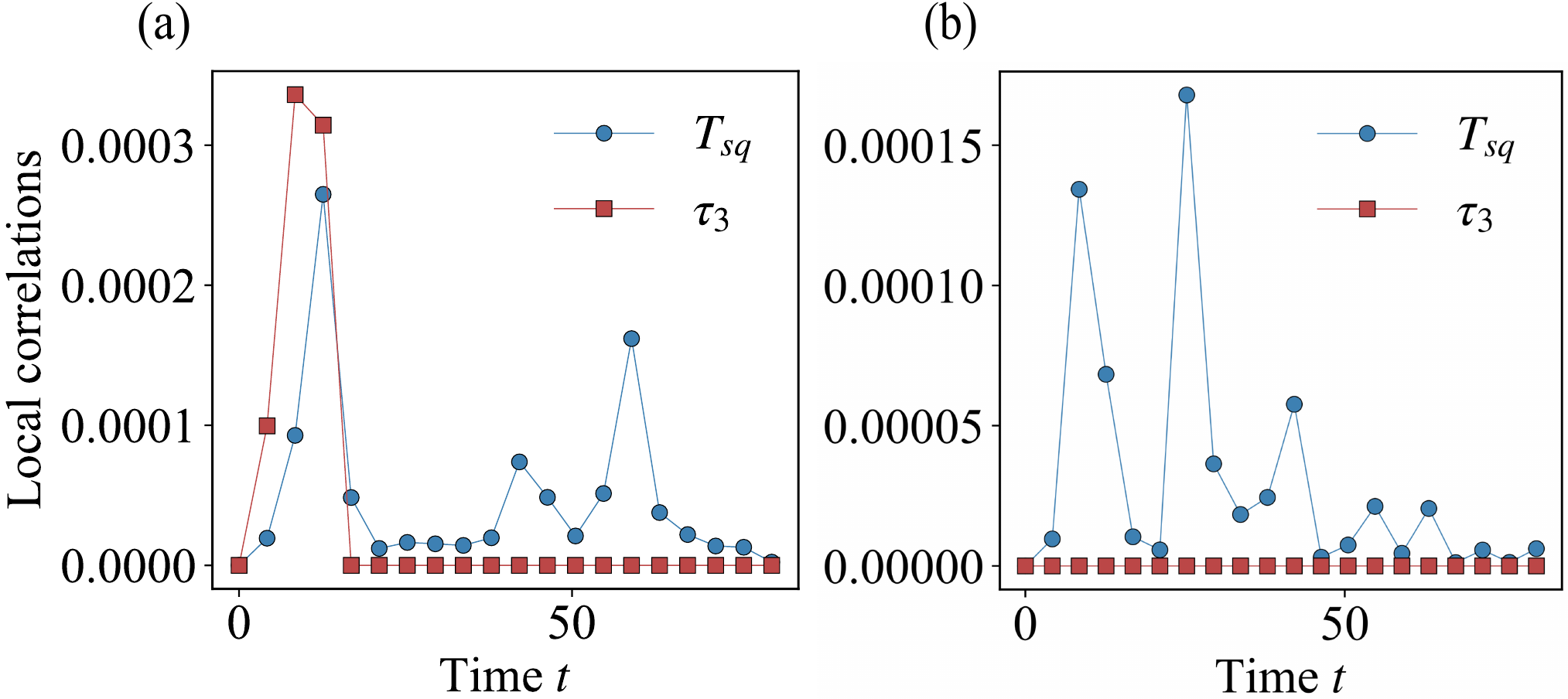}
    \caption{(a) For adjacent qubits, the quench induces local GHZ-type GME that quickly collapses due to $\gamma$ and monogamy as correlations spread. In contrast, $T_{\sq}$ remains positive, witnessing bipartite and GME. (b) For distant qubits, $\tau_{3} = 0$ as the correlations, mediated by quasiparticles, are bipartite. Under the quench, $T_{\sq}$ oscillates due to coherent recurrences in the non-local regime. In both scenarios, $T_{\sq}$ captures genuinely quantum correlations that $\tau_{3}$ misses.}
    \label{fig:fig2}
\end{figure}

\textit{Discussion.---} We introduced the hysteretic squashed entanglement \(T_{\sq}\) as a measure of irreducible conditional quantum correlations in many-body systems. By minimizing the QCMI over compatible state extensions, \(T_{\sq}\) removes correlations mediated through the conditioning region and hidden environments, isolating those that cannot be reproduced by extension-based recovery, Fig.~\ref{fig:fig1}. Its properties (Proposition~\ref{thm:prop}) establish it as a consistent conditional entanglement monotone; its ordering with $E_{\sq}$ and $N_{\sq}$ (Lemma~\ref{lem:lower_bounds}) places it within a hierarchy of squashed correlations, and it is proven to be upper bounded by co(QCMI) in Theorem~\ref{thm:thm1}. $T_{\sq}=0$ characterizes conditional recoverability: all conditioned $A$-$C$ correlations are compatible with reconstruction from local data and a suitable state extension, whereas $T_{\sq}>0$  quantifies the irreducible obstruction to such reconstruction. 

We further observed this in Fig.~\ref{fig:fig2}, where compared with $\tau_3$, $T_{\sq}$ detects conditional correlations in mixed, non-equilibrium regimes where residual GME measures fail. This viewpoint provides a quantitative route to mixed-state topological structure (Theorem~\ref{thm:thm2}). Here, $T_{\sq}>0$ signals conditionally protected long-range quantum correlations that cannot be locally mediated, while $T_{\sq}=0$ indicates that apparent long-range correlations are reconstructible and do not encode intrinsic global order. These results open immediate avenues for studying resource theory of conditional entanglement based on $T_{\sq}$ and using $T_{\sq}$ to classify mixed-state phases; particularly to characterize the emergence, stability, and dynamical transitions of topological order beyond equilibrium and pure states.

\textit{Acknowledgments.---} S.D. and A.Y. contributed equally and are listed alphabetically. S.D. acknowledges support from the ANRF, DST, Govt.~of India, under Grant No.~SRG/2023/000217, MeitY, Govt.~of India, under Grant No.~4(3)/2024-ITEA, and the IIIT Faculty Seed Grant.

\section*{Appendix}

\phantomsection
\textit{Appendix A: Quantum conditional mutual information.---} Here we provide some further details about QCMI. $I(A;C|B)_{\rho}=0$ iff there exists a universal recovery map $\mc{R}_{B\to BC}$, such that $\mc{R}_{B\to BC}(\rho_{AB})=\rho_{ABC}$~\cite{Pet88}. In fact, for $I(A;C|B)_{\rho}=0$, the recovery map $\mc{R}_{B\to BC}$ coincides with the Petz recovery map $\mc{R}^{\rm P}_{B\to BC}$~\cite{Pet88,FR15,JRS+18}
\begin{equation}
    \mathcal{R}^{\mathrm P}_{B\to BC}(\cdot)=
\rho_{BC}^{1/2}
\left(
\rho_B^{-1/2} \, (\cdot) \, \rho_B^{-1/2}
\otimes \mathbbm{1}_C
\right)
\rho_{BC}^{1/2}.
\end{equation}
For a state $\rho_{AB}$, $\rho_A=\tr_B(\rho_{AB})$ is a reduced or marginal state, and any state $\omega_{ABC}$ such that $\tr_C(\omega_{ABC})=\rho_{AB}$ is called its state extension. For finite-dimensional systems, a state $\rho_{AC}$ is separable (not entangled) if there exists a state extension $\rho_{ACE}$ such that $I(A;C|E)_{\rho}=0$~\cite{LW18}. In other words, a quantum state $\rho_{AC}$ is not entangled if there exists a state extension $\rho_{ACE}$ and a universal recovery map $\mc{R}_{E\to EC}$ such that $\mc{R}_{E\to EC}(\rho_{AE})=\rho_{ACE}$. As highlighted in the main text, this notion was recently extended to quantify the squashed quantum non-Markovianity $N_{\sq}$~\cite{GPG+25} that captures entanglement between $A$ and $C$, conditioned on $B$, for a tripartite state $\rho_{ABC}$ if there exists a state extension $\rho_{ABCE}$, such that $I(A;C|BE)_{\rho}=0$. This makes sense because genuine quantum correlations between two systems are not shareable with infinitely many systems~\cite{KW04,KDWW19}. If $I(A;C|B)_{\rho}=0$ then one can sequentially apply $\mc{R}_{B\to BC}(\rho_{AB})$ many times to obtain $\mathcal{R}\circ\cdots\circ\mathcal{R}(\rho_{AB})
= \rho_{ABC_1\ldots C_k}$, where $C_i\simeq C$ for $i\in\{1,2,\ldots, k\}$ such that $\tr_{C_1\ldots C_k\setminus C_i}=\rho_{AC}$ $\forall k\in\mathbb{N}$~\cite{LW18}, which is possible iff $\rho_{AC}$ is separable~\cite{DPS04}. 

\textit{Appendix B: co(QCMI) and $T_{\sq}$.---} For a mixed state $\rho_{ABCD}$, the co(QCMI)~\cite{WSMG25} is formally defined as
\begin{equation}\label{eq:coqcmi}
    {\rm co(QCMI)[\rho_{ABCD}]}:= \frac{1}{2}\inf_{{\substack{\{p_i,\psi^i\}_i: \\ \rho=\sum_ip_i\psi^i}}}\sum_{i}p_iI(A:C|B)_{\psi^i},
\end{equation}
where the infimum is taken over all ensembles $\{p_i,\psi^i_{ABCD}\}_i$ of pure states. For a pure state $\psi_{ABCD}$, the upper bound in Theorem~\ref{thm:thm1} is saturated 
\begin{equation}
T_{\sq}(A;C|B)_{\psi}={\rm co(QCMI)[\psi_{ABCD}]}=\frac{1}{2}I(A;C|B)_{\psi}=\gamma,
\end{equation}
which can be expected since every extension of a pure state is trivial, so the optimization in $T_{\sq}$ is attained by the purification itself.

\phantomsection
\textit{Appendix C: Proofs of main results.---}\label{app:C} We now prove Theorem~\ref{thm:thm1} and Theorem~\ref{thm:thm2}.\\

\textbf{Theorem 1.} \textit{For an arbitrary relevant state $\rho_{ABCD}$, the hysteretic TEE $T_{\sq}(A;C|B)_{\rho}$ is upper bounded by the co(QCMI)
\begin{equation}
   0\leq T_{\sq}(A;C|B)_{\rho}\leq {\rm co(QCMI)}[\rho_{ABCD}]. \nonumber
\end{equation}}

\begin{proof}
The hysteretic TEE $T_{\sq}$ is nonnegative because of the nonnegativity of the conditional mutual information of a quantum state. For a relevant state $\rho_{ABCD}$, let $\{p_i,\psi^i_{ABCD}\}_i$ denote an ensemble of pure states that gives optimal value for ${\rm co(QCMI)}[\rho_{ABCD}]$. Using this optimal ensemble, we can form a state extension $\sigma_{ABCDE}$ of $\rho_{ABCD}$ such that
\begin{equation}
\sigma_{ABCDE}=\sum_{i}p_i\psi^i_{ABCD}\otimes\op{i}_E,
\end{equation}
where $\{\ket{i}\}_i$ is an orthonormal set of vectors. It follows that
\begin{align}
T_{\sq}(A;C|B)_{\rho} &\le \frac{1}{2} I(A;C|BE)_{\sigma} \nonumber\\
&= \frac{1}{2}\sum_{i} p_i I(A;C|B)_{\psi^i} \nonumber\\
&= \mathrm{co(QCMI)}[\rho_{ABCD}],
\end{align}
where, to arrive at the first equality, we used the fact that
\begin{equation}
S(XA)_{\rho}=H(X)+\sum_xp_X(x)S(\rho^x_A) 
\end{equation}
for a classical-quantum state 

\begin{equation}
\rho_{XA}=\sum_xp_X(x)\op{x}_X\otimes\rho^x_A.
\end{equation}
This proves the upper bound on $T_{\sq}(A;C|B)_{\rho}$, which is saturated if $\rho_{ABCD}$ is pure because the extending system $E$ can only be in product state with $ABCD$, and $I(A;C|B)_{\psi}=I(A;C|BE)_{\psi}$ for all possible state extensions $\psi_{ABCDE}$ of an arbitrary pure state $\psi_{ABCD}$.\\
\end{proof}
     
\textbf{Theorem 2.} \textit{If $T_{\sq}(A;B|C)_{\rho}$ of a quantum state $\rho_{ABCD}$ changes negligibly under the action of a quantum channel $\mc{E}_{A\to A'}$ on $A$ then there exists a state extension $\rho_{ABCDE}$ and a universal recovery map $\mc{R}_{\rho_{ABE}\otimes\rho_C,\mc{E}}$ that almost recovers the state $\rho_{ABCE}$ from $\mc{E}(\rho_{ABCE})$. In particular,
     \begin{equation}
         T_{\sq}(A;C|B)_{\rho}-T_{\sq}(A';C|B)_{\mc{E}(\rho)}\leq \varepsilon  \nonumber
     \end{equation}
     implies the existence of a quantum channel $\mc{R}_{\rho_{ABE}\otimes\rho_C,\mc{E}}$ such that 
     \begin{equation}
        \sup_{\substack{\rho_{ABCE}:\\ \tr_E(\rho_{ABCDE})=\rho_{ABCD}}} F(\rho_{ABCE},\mc{R}\circ\mc{E}(\rho_{ABCE}))\geq 2^{-\varepsilon}, \nonumber
     \end{equation}
     and $\mc{R}_{\rho_{ABE}\otimes\rho_C,\mc{E}}\circ\mc{E}(\rho_{ABCE})\approx\rho_{ABCE}$ iff
     $\varepsilon\approx 0$.}
\begin{proof}

    For any state extension $\rho_{ABCDE}$ of $\rho_{ABCD}$ and any quantum channel $\mc{E}_{A\to A'}$, there exists a universal recovery map $\mc{R}$ such that
\begin{align}
&I(A;C|BE)_{\rho} - I(A;C|BE)_{\mc{E}(\rho)}\nonumber\\
&=I(ABE;C)_{\rho}-I(BE;C)_{\rho}
   - I(ABE;C)_{\mc{E}(\rho)}\nonumber\\
   &\qquad\qquad\qquad\qquad \qquad +I(BE;C)_{\mc{E}(\rho)} \nonumber\\
&= I(ABE;C)_{\rho} - I(ABE;C)_{\mc{E}(\rho)} \nonumber\\
&\ge - \log F\!\left(
\rho_{ABCE},
\mc{R}_{\rho_{ABE}\otimes\rho_{C},\mc{E}}
\circ \mc{E}(\rho_{ABCE})
\right),
\end{align}
where the last inequality follows from the strengthened data-processing inequality~\cite{JRS+18} after recognizing that the QMI of a quantum state is equal to the quantum relative entropy between the states and the tensor-product of its reduced state 
\begin{equation}
I(A;B)_{\sigma}=D(\sigma_{AB}\|\sigma_A\otimes\sigma_B).
\end{equation}
In particular, for arbitrary quantum states $\rho_A,\sigma_A$ and arbitrary quantum channel $\mathcal{N}_{A\to B}$, there exists a universal recovery map $\mathcal{R}_{\sigma,\mathcal{N}}$ such that~\cite{JRS+18}
\begin{equation}
    D(\rho\|\sigma)-D(\mathcal{N}(\rho)\|\mathcal{N}(\sigma))\geq -\log F(\rho,\mathcal{R}_{\sigma,\mathcal{N}}\circ\mathcal{N}(\rho)).
\end{equation}
Furthermore, $ D(\rho\|\sigma)=D(\mathcal{N}(\rho)\|\mathcal{N}(\sigma))$ iff~\cite{Pet88}
\begin{equation}
    \mathcal{R}_{\sigma,\mathcal{N}}(\cdot)=\sigma^{\frac{1}{2}}\mathcal{N}^\dag[\mathcal{N}(\sigma)^{-\frac{1}{2}}(\cdot)\mathcal{N}(\sigma)^{-\frac{1}{2}}]\sigma^{\frac{1}{2}}.
\end{equation}

Taking infimum over all state extensions $\rho_{ABCDE}$ of $\rho_{ABCD}$, we get
   \begin{align}
&T_{\sq}(A;C|B)_{\rho}
 - T_{\sq}(A';C|B)_{\mc{E}(\rho)}\nonumber\\
&\ge \inf_{\rho_{ABCDE}}
\!\left[
I(A;C|BE)_{\rho}
 - I(A;C|BE)_{\mc{E}(\rho)}
\right] \nonumber\\
&\ge - \log
\sup_{\rho_{ABCDE}}
F\!\left(
\rho_{ABCE},
\mc{R}_{\rho_{ABE}\otimes\rho_C,\mc{E}}
\circ \mc{E}(\rho_{ABCE})
\right).
\end{align}
We know that $0\leq F(\tau,\sigma)\leq 1$ for any two states $\tau_A$ and $\sigma_A$; $F(\tau,\sigma)=1$ iff $\tau=\sigma$, and $F(\tau,\sigma)=0$ iff $\tau\perp\sigma$. We have 

\begin{equation}
1 - \sqrt{F(\tau,\sigma)} \;\le\; \frac{1}{2}\|\rho - \sigma\|_1
\;\le\; \sqrt{\,1 - F(\rho,\sigma)\,}.
\end{equation}
This observation concludes the proof as $F(\tau,\sigma)\approx 1$ iff $\tau\approx \sigma$.
\end{proof}

Theorem~\ref{thm:thm2} formalizes that small change of $T_{\sq}$ under local processing implies approximate recoverability of the global correlations conditioned on the mediating region. This connects irreducible conditional entanglement to robustness of information flow under local noise.\\

\textit{Appendix D: Extension to dynamical settings.---} Our formalism of hysteretic entanglement can be extended to study dynamical memory effects in open many-body quantum systems. This perspective is motivated by the application of $N_{\sq}$ in the context of revivals of information in open quantum systems~\cite{BGG+25}.

Let $\rho_{ABCD}(t)$ denote the evolved state at time $t=t$ with $t=0$ denoting the initial time. Consider that the evolution is happening under quantum dynamics with Liouvillian $\mc{L}^{(k)}$ for $k\in\{1,2,\ldots\}$ and we do not know under which dynamics the state evolves. Besides, the dynamics $\mc{L}^{(k)}$ are such that if $\rho_{ABCD}(0)$ has $I(A;C|B)_{\rho(0)}=0$, then $I(A;C|B)_{{\rm e}^{\mc{L}^{(k)}_t}(\rho)}=0$. That is, the dynamics $\mc{L}^{(k)}$ preserves the quantum Markov chain property $A$-$B$-$C$ of the initial state. In this case, there exists state $\rho_{ABCD}(0)$ and dynamics $\mc{L}^{(k)}$ for $k\in\{1,2,\ldots\}$ such that when we do not know which $\mc{L}^{(k)}$ occurred, we could have $I(A;C|B)_{\rho(t)}>0$ and still we will always have $T_{\sq}(A;C|B)_{\rho(t)}=0$. This holds because of the convexity of $T_{\sq}$, see Proposition~\ref{thm:prop}.

Let us now consider a tripartite state $\rho_{ABC}(t_0)=\Phi_{AC}\otimes\gamma_B$ at $t=t_0$, where $\Phi_{AC}:=\sum_{i,j}\ket{ii}\bra{jj}_{AC}$ is a maximally entangled state and $\gamma_B$ is some state. Suppose the state $\rho_{ABC}$ undergoes a two-step unitary transformation: $\mc{U}_{AB\to A_1B_1}$ from $t=t_0$ to $t=t_1$ and $\mc{V}_{A_1B_1\to A_2B_2}$ from $t=t_1$ to $t=t_2$, i.e., $\rho_{ABC}{(t_0)}\to \rho_{A_1B_1C}{(t_1)}\to \rho_{A_2B_2C}{(t_2)}$. Such an open quantum system dynamics is said to exhibit a non-causal revival of information~\cite{BGG+25}, whenever a revival occurs $I(A_2;C)_{\rho(t_2)}>I(A_1;C)_{\rho(t_1)}$, but there exists a compatible model for the dynamical evolution and an inert extension system $E$ such that $I(A_2E;C)_{\rho(t_2)}\leq I(A_1E;C)_{\rho(t_1)}$. If $E_{\sq}(A_2;C)_{\rho(t_2)}>S(A_1)_{\rho(t_1)}$, then the revival requires a non-zero amount of genuine backflow to be explained~\cite{BGG+25}.

\newpage

\widetext
\section*{Supplemental Materials}
\phantomsection
\label{sec:SM}

We ought to start with Ref. \cite{GPG+25} which formulated the notion of squashed quantum non-Markovianity $N_{\sq}$. The properties of $T_{\sq}$, as defined in Proposition~\ref{thm:prop} in the main text, were proven in Lemma 2 of \cite{GPG+25}. The squashed-based measures $E_{\sq}$, $N_{\sq}$, and $T_{\sq}$ form a hierarchy distinguished by the amount of conditioning allowed when removing mediated correlations. Where $E_{\sq}$ removes correlations reproducible via an external extension, $N_{\sq}$ additionally allows conditioning on a mediating subsystem, whereas $T_{\sq}$ allows mediation in the case of embedding within a larger environment. With that, we now prove Lemma~\ref{lem:lower_bounds}.\\

\textbf{Lemma 1.} \textit{Suppose $\rho_{AC}$ and $\rho_{ABC}$ are reduced states of $\rho_{ABCD}$. The squashed quantum correlations $E_{\sq}$ of $\rho_{AC}$, $N_{\sq}$ of $\rho_{ABC}$, and $T_{\sq}$ of $\rho_{ABCD}$ can then be ordered as}
\begin{align}
E_{\sq}(A;C)_{\rho}\leq N_{\sq}(A;C|B)_{\rho}
\leq T_{\sq}(A;C|B)_{\rho} \nonumber
\leq \frac{1}{2}I(A;C|D)_{\rho}. \nonumber
\end{align}

\begin{proof}
The inequality $E_{\sq}(A;C)_{\rho}\leq N_{\sq}(A;C|B)_{\rho}$ was proven in Lemma 1 of \cite{GPG+25}. Furthermore, $N_{\sq}(A;C|B)_{\rho}\leq T_{\sq}(A;C|B)_{\rho}$ follows directly from Definition 1 of \cite{GPG+25} and Definition~\ref{def:df1} in this work. We notice that $\rho_{ABC}$ is a marginal state of $\rho_{ABCD}$, then the set of all possible state extensions $\sigma^2_{ABCDE}$ of $\rho_{ABCD}$ is contained in the set of all possible state extensions $\sigma^1_{ABCE}$ of $\rho_{ABC}$, and
\begin{align}
N_{\sq}(A;C|B)_{\rho}
&= \frac{1}{2}\inf_{\sigma^1_{ABCE}} I(A;C|BE)_{\sigma^1} \nonumber\\
&\le \frac{1}{2}\inf_{\sigma^2_{ABCDE}} I(A;C|BE)_{\sigma^2} \nonumber\\
&= T_{\sq}(A;C|B)_{\rho}.
\end{align}
We get $T_{\sq}(A;C|B)_{\rho}\leq \frac{1}{2}I(A;C|D)_{\rho}$ by considering purification $\psi^{\rho}_{ABCDE}$ of a state $\rho_{ABCD}$, i.e., $\psi^{\rho}_{ABCDE}$ is a pure state such that $\tr_E(\psi^{\rho}_{ABCDE})=\rho_{ABCD}$
\begin{align}
T_{\sq}(A;C|B)_{\rho} \le \frac{1}{2}I(A;C|BE)_{\psi^{\rho}} \nonumber
&= \frac{1}{2}I(A;C|D)_{\psi^{\rho}} = \frac{1}{2}I(A;C|D)_{\rho},
\end{align}
where the first equality follows from the duality of the QCMI for pure states: for any pure state $\phi_{ABCD}$, we have
\begin{equation}
I(A;C|B)_{\phi}=I(A;C|D)_{\phi}.
\end{equation}

\end{proof}

The bound $T_{\sq}(A;C|B)_\rho \le \frac{1}{2}  I(A;C|D)_\rho$ shows that irreducible conditional entanglement cannot exceed the total conditional correlations across any complementary partition. Thus, $T_{\sq}$ isolates the portion of QCMI after squashing out all extension-assisted mediation.

Exact invariance of $T_{\sq}$ under a local channel characterizes a quantum Markov structure conditioned on the mediating subsystem and compatible extensions. Thus, vanishing change of $T_{\sq}$ identifies states whose conditional correlations are entirely reconstructible. The following observation directly follows from Theorem~\ref{thm:thm2} and the saturation condition of the data-processing inequality~\cite{Pet88,FR15,JRS+18,sohail2025fundamental}.
\begin{corollary}
     The hysteretic TEE $T_{\sq}(A;C|B)_{\rho}$ of a quantum state $\rho_{ABCD}$ remains invariant under the action of a quantum channel $\mc{E}_{A\to A'}$ on $A$ iff there exists a state extension $\rho_{ABCDE}$ and a universal recovery map $\mc{R}_{\rho_{ABE}\otimes\rho_C,\mc{E}}$ that perfectly recovers the state $\rho_{ABCE}$ from $\mc{E}(\rho_{ABCE})$. That is
     \begin{equation}
         T_{\sq}(A;C|B)_{\rho}= T_{\sq}(A';C|B)_{\mc{E}(\rho)}
     \end{equation}
    iff there exists of a state extension $\rho_{ABCDE}$ and a universal recovery map $\mc{R}_{\rho_{ABE}\otimes\rho_C,\mc{E}}$ such that 
     \begin{equation}
        \sup_{\substack{\rho_{ABCE}:\\ \tr_E(\rho_{ABCDE})=\rho_{ABCD}}} F(\rho_{ABCE},\mc{R}\circ\mc{E}(\rho_{ABCE}))= 1.
     \end{equation}
Furthermore, $\mc{R}_{\rho_{ABE}\otimes\rho_C}$ reduces to the Petz recovery map $\mc{R}^{\rm P}_{\rho_{ABE}\otimes\rho_C}$~\cite{Pet88}
\begin{equation}
\mc{R}^{\rm P}_{\sigma,\mc{N}}(\cdot):= \sigma^{1/2}\mc{N}^\dag[\mc{N}(\sigma)^{-1/2}(\cdot)\mc{N}(\sigma)^{-1/2}]\sigma^{1/2}. 
   \end{equation}
\end{corollary}
It follows that the Petz recovery map $\mc{R}^{\rm P}_{\sigma,\mc{N}}$, for $\sigma=\rho_{ABE}\otimes{\rho_C}$ and $\mc{N}=\mc{E}_{A\to A'}$, is of the form
\begin{align*}
    & \mc{R}_{\rho_{ABE}\otimes\rho_C,\mc{E}_{A\to A'}}\nonumber
     = \Pi^{\rho}_{C}\rho_{ABE}^{1/2}\mc{E}^\dag[\mc{E}(\rho_{ABE})^{-1/2}(\cdot)\mc{E}(\rho_{ABE})^{-1/2}]\rho_{ABE}^{1/2}\Pi^{\rho}_{C},
\end{align*}
where $\Pi^\rho_C$ is projector on the support of the reduced state $\rho_C$.

\end{document}